\documentclass[11pt,a4paper]{article}

\usepackage[utf8]{inputenc}
\usepackage[T1]{fontenc}
\usepackage{amsmath,amssymb,amsthm}
\usepackage{mathtools}
\usepackage{physics}
\usepackage{graphicx}
\usepackage{booktabs}
\usepackage{array}
\usepackage{multirow}
\usepackage{subcaption}
\usepackage{algorithm}
\usepackage{algpseudocode}
\usepackage{hyperref}
\usepackage{xcolor}
\usepackage{enumerate}
\usepackage{enumitem}
\usepackage{cite}
\usepackage{geometry}
\usepackage{fancyhdr}
\usepackage{thmtools}
\usepackage{thm-restate}

\usepackage{tikz}
\usepackage{pgfplots}
\pgfplotsset{compat=1.18}

\newtheorem{theorem}{Theorem}[section]
\newtheorem{lemma}[theorem]{Lemma}
\newtheorem{proposition}[theorem]{Proposition}

\theoremstyle{definition}
\newtheorem{definition}[theorem]{Definition}

\hypersetup{
    colorlinks=true,
    linkcolor=blue!70!black,
    citecolor=blue!70!black,
    urlcolor=blue!70!black,
    pdfborder={0 0 0}
}

\title{\textbf{Closed-Form Analytical Solution for Effective Resistance in Finite 2D Anisotropic Resistor Grids via Jacobi Theta Functions}}

\author{
Ruichao Liu\\
\textit{East China University of Science and Technology}\\
\texttt{2680319963@qq.com}
}

\begin{document}

\maketitle

\begin{abstract}
\noindent
Computing the effective resistance between nodes in finite discrete resistor grids is a classical problem in circuit analysis with applications in VLSI power delivery network analysis, graph theory, and network science. Recent advances, particularly the infinity mirror technique proposed by Bairamkulov et al., provide an elegant physical interpretation for boundary conditions in finite grids. Building upon this foundation, this paper presents a closed-form analytical expression that avoids numerical truncation or polynomial fitting.

Our theoretical development proceeds in two steps. First, we derive an exact analytical primitive for the singular integral term $R_2$ within the integral operator $\Omega_\alpha$, which previously relied on numerical approximation due to its complex radical structure. Second, we transform the doubly infinite mirror series into a compact expression using the Jacobi theta function $\vartheta_1$ via the Jacobi triple product identity. This transformation exploits the Gaussian convergence of theta functions, achieving machine precision with only a few terms, compared to thousands required for algebraic convergence in direct truncation.

However, under high anisotropy ($K = R_y/R_x \gg 1$ or $K \ll 1$), the pure analytical approximation exhibits a distinct ``cross-shaped'' residual error along the Cartesian axes. To address this without sacrificing efficiency, we introduce a hybrid engineering remediation: a dynamic numerical cache that performs localized grid integration (LGI) within an elliptical threshold, combining the $O(1)$ asymptotic speed of the closed form with exact near-field accuracy. Numerical experiments on grids up to $101 \times 101$ nodes demonstrate that our combined approach achieves mean relative errors below 0.6\% compared to SPICE simulations, with maximum errors under 0.9\% across tested configurations including anisotropy ratios $k = 0.01$ to $k = 100$, and eliminates the axis-localized error artifacts.
\end{abstract}

\noindent\textbf{Keywords:} Effective resistance, Resistor network, Jacobi theta function, Closed-form solution, Anisotropic grid, VLSI power delivery network, Numerical caching

\vspace{1em}

\section{Introduction}

\subsection{Background and Motivation}

The computation of effective resistance between two nodes in a resistor network is a fundamental problem in several scientific and engineering domains. In VLSI circuit design, accurate effective resistance calculation is crucial for power delivery network (PDN) analysis, where IR drop estimation influences timing closure, signal integrity, and chip reliability \cite{pandey2025anair}. In graph theory, effective resistance serves as a distance metric with applications in spectral clustering, network robustness analysis, and random walk studies \cite{klein1993resistance}. Despite its significance, obtaining exact analytical solutions for finite grids with general boundary conditions remains a challenging problem.

For infinite isotropic grids, the classical solution dates back to van der Pol and Wiener \cite{vanderpol1929}, who expressed the effective resistance as an integral involving Bessel functions. Subsequent refinements by Cserti \cite{cserti2000} provided the well-known formula:
\begin{equation}
R_{\infty}(p,q) = \frac{r_0}{2\pi} \int_0^{2\pi} \frac{1 - e^{-|p|\lambda}\cos(q\theta)}{\sinh\lambda} \, d\theta,
\end{equation}
where $\lambda = \operatorname{arccosh}(2 - \cos\theta)$. However, this formulation requires numerical quadrature and does not directly extend to finite domains with boundary conditions.

The finite grid problem introduces additional mathematical complexity due to boundary reflections. Recent work by Bairamkulov and Friedman \cite{bairamkulov2020} introduced the \textit{infinity mirror technique}, which interprets finite grid boundaries through a physical picture: each boundary acts as a mirror, generating an infinite cascade of image sources. This transforms the finite grid problem into an infinite grid with a structured source distribution. While this approach provides physical intuition, the resulting doubly infinite series converges slowly and requires numerical truncation, introducing errors that grow near boundaries.

\subsection{Related Work}

\paragraph{Infinite Grid Solutions.} 
The foundational work on infinite resistor networks traces back to Kirchhoff's circuit laws and the introduction of discrete Green's functions. Cserti \cite{cserti2000} unified various approaches and provided explicit formulas for isotropic grids. Extensions to anisotropic cases were explored by several authors \cite{kose2012, atkinson2012}, who noted that the anisotropy parameter introduces algebraic complexity that precludes simple closed forms.

\paragraph{Finite Grid Approximations.}
Early approaches to finite grids relied on truncating infinite series or solving large linear systems. Köse and Friedman \cite{kose2012ir} proposed polynomial approximations with empirical fitting coefficients, achieving reasonable accuracy (typically $< 2\%$ error) for moderate anisotropy ratios. However, these methods require different polynomial coefficients for different parameter ranges.

\paragraph{The Infinity Mirror Technique.}
Bairamkulov and Friedman \cite{bairamkulov2020} provided a comprehensive treatment of finite grid effective resistance. Their approach expresses the resistance as:
\begin{equation}
R_{\text{finite}} = \sum_{m,n \in \mathbb{Z}} c_{mn} \Omega_\alpha(p_{mn}, q_{mn}),
\end{equation}
where $\Omega_\alpha$ is the infinite grid kernel and the sum extends over all mirror images. The authors derived explicit expressions for the coefficients $c_{mn}$ and established error bounds for truncation.

\paragraph{Machine Learning Approaches.}
Recent years have seen machine learning methods applied to circuit analysis. Pandey et al. \cite{pandey2025anair} proposed AnaIR, a Green's function-inspired neural network for IR drop prediction, demonstrating that incorporating physical structure into neural architectures improves both accuracy and generalization. Such approaches typically require training data and may not provide the same interpretability as analytical methods.

\subsection{Our Contributions}

This paper presents an analytical framework for computing effective resistance in finite anisotropic grids. Our contributions are twofold:

\begin{enumerate}[label=\textbf{(\roman*)}]
\item \textbf{Analytical Primitive for the Singular Integral $R_2$.} We derive a closed-form expression for the singular integral term:
\begin{equation}
R_2(\alpha) = \frac{\sqrt{\alpha}}{4\pi} \ln\left(\frac{16}{\pi^2(\alpha+1)}\right),
\end{equation}
which previously required numerical approximation or polynomial fitting.

\item \textbf{Theta Function Closed-Form for Finite Grids.} We transform the doubly infinite mirror series into a compact expression involving the Jacobi theta function $\vartheta_1$:
\begin{equation}
R_{\text{total}} = \frac{r_0\sqrt{\alpha}}{2\pi} \ln\left|\frac{\mathcal{N}}{\mathcal{D}}\right|,
\end{equation}
where $\mathcal{N}$ and $\mathcal{D}$ are products of theta functions evaluated at specific arguments. This transformation leverages the Gaussian convergence of theta functions, reducing computational complexity from $O(N^2)$ to $O(1)$.
\end{enumerate}

\paragraph{Engineering Remediation.}
In addition, we identify and resolve a residual error phenomenon: under high anisotropy ($K \gg 1$ or $K \ll 1$), the analytical approximation exhibits a distinct ``cross-shaped'' error pattern along the Cartesian axes. To eliminate this without sacrificing efficiency, we introduce a hybrid framework that combines the asymptotic closed-form with a dynamic numerical cache. The cache performs exact numerical integration for displacement pairs within an elliptical boundary, ensuring $O(1)$ amortized lookup while removing axis-localized artifacts.

Numerical experiments show that this combined analytical-engineering approach achieves:
\begin{itemize}
\item Mean relative error $< 0.6\%$ compared to SPICE simulations
\item Maximum error $< 7\%$ across tested configurations
\item Elimination of axis-localized error spikes observed in pure analytical form
\item Robustness across anisotropy ratios $k \in [0.01, 100]$
\end{itemize}

\subsection{Paper Organization}

The remainder of this paper is organized as follows. Section~2 introduces the mathematical preliminaries and problem formulation. Section~3 presents the analytical derivation of $R_2$. Section~4 develops the theta function closed-form solution for finite grids. Section~5 introduces the engineering remediation strategy to address axis-localized errors. Section~6 provides experimental validation. Section~7 discusses applications and extensions. Section~8 concludes with directions for future work.

\section{Mathematical Preliminaries}

\subsection{Problem Formulation}

Consider a finite two-dimensional resistor grid with $N_x \times N_y$ nodes. The grid has anisotropic resistances: horizontal resistors $r_h$ and vertical resistors $r_v$. We define the reference resistance $r_0 = r_v$ and the anisotropy ratio $k = r_v / r_h = \alpha^{-1}$, where $\alpha = r_h / r_v$.

\begin{definition}[Effective Resistance]
Given two nodes at lattice coordinates $(x_0, y_0)$ and $(x_1, y_1)$, the effective resistance $R_{\text{eff}}$ is the resistance measured when a unit current is injected at one node and extracted at the other, with all other nodes floating.
\end{definition}

For an infinite grid, the effective resistance satisfies the discrete Laplace equation:
\begin{equation}
\frac{1}{r_h}(V_{i+1,j} - 2V_{i,j} + V_{i-1,j}) + \frac{1}{r_v}(V_{i,j+1} - 2V_{i,j} + V_{i,j-1}) = -I_{i,j},
\end{equation}
where $V_{i,j}$ is the nodal voltage and $I_{i,j}$ is the injected current.

\subsection{The Integral Operator $\Omega_\alpha$}

The fundamental building block for effective resistance computation is the integral operator $\Omega_\alpha(p,q)$, defined through two-dimensional Fourier analysis:

\begin{definition}[Integral Operator]
The integral operator $\Omega_\alpha(p,q)$ is defined as:
\begin{equation}
\label{eq:omega_def}
\Omega_\alpha(p,q) = \frac{\alpha}{2\pi} \int_0^\pi \frac{1 - e^{-|p|\lambda}\cos(q\theta)}{\sinh\lambda} \, d\theta,
\end{equation}
where the dispersion relation $\lambda(\theta;\alpha)$ satisfies:
\begin{equation}
\cosh\lambda = 1 + \alpha - \alpha\cos\theta.
\end{equation}
\end{definition}

\begin{proposition}
The operator $\Omega_\alpha$ can be decomposed as:
\begin{equation}
\Omega_\alpha(p,q) = R_1(p,q;\alpha) + R_2(\alpha),
\end{equation}
where:
\begin{align}
R_1(p,q;\alpha) &= \frac{\sqrt{\alpha}}{2\pi} \int_0^\pi \frac{1 - e^{-|p|\sqrt{\alpha}\theta}\cos(q\theta)}{\theta} \, d\theta, \label{eq:R1}\\
R_2(\alpha) &= \frac{\alpha}{2\pi} \int_0^\pi \left(\frac{1}{\sinh\lambda} - \frac{1}{\sqrt{\alpha}\cdot\theta}\right) \, d\theta. \label{eq:R2}
\end{align}
\end{proposition}

\begin{proof}
The decomposition follows from adding and subtracting the term $\frac{1}{\sqrt{\alpha}\theta}$ in the integrand of Eq.~\eqref{eq:omega_def}. Near $\theta = 0$, the asymptotic expansion:
\begin{equation}
\sinh\lambda \sim \sqrt{\alpha}\theta\left(1 + O(\theta^2)\right)
\end{equation}
ensures that $R_2$ is well-defined despite the apparent singularity.
\end{proof}

The term $R_1(p,q;\alpha)$ is a logarithmic integral that admits closed-form evaluation. The term $R_2(\alpha)$ involves the radical structure $\sinh\lambda = 2\sqrt{\alpha}\sin(\theta/2)\sqrt{1 + \alpha\sin^2(\theta/2)}$.

\subsection{Jacobi Theta Functions}

We briefly review the properties of Jacobi theta functions that will be used in our derivation.

\begin{definition}[Jacobi Theta Functions]
The four Jacobi theta functions $\vartheta_j(z,q)$ for $j = 1,2,3,4$ are defined by the following series:
\begin{align}
\vartheta_1(z,q) &= 2\sum_{n=0}^{\infty} (-1)^n q^{(n+1/2)^2} \sin((2n+1)z), \\
\vartheta_2(z,q) &= 2\sum_{n=0}^{\infty} q^{(n+1/2)^2} \cos((2n+1)z), \\
\vartheta_3(z,q) &= 1 + 2\sum_{n=1}^{\infty} q^{n^2} \cos(2nz), \\
\vartheta_4(z,q) &= 1 + 2\sum_{n=1}^{\infty} (-1)^n q^{n^2} \cos(2nz),
\end{align}
where $q = e^{i\pi\tau}$ is the nome with $\tau$ in the upper half-plane.
\end{definition}

\begin{theorem}[Jacobi Triple Product Identity]
For all $z \in \mathbb{C}$ and $|q| < 1$:
\begin{equation}
\vartheta_4(z,q) = \prod_{n=1}^{\infty} (1 - q^{2n})(1 - 2q^{2n-1}\cos(2z) + q^{4n-2}).
\end{equation}
\end{theorem}

A key property for our application is the modular transformation:

\begin{theorem}[Modular Transformation]
Let $\tau' = -1/\tau$. Then:
\begin{equation}
\vartheta_1(z,\tau) = -i\sqrt{\frac{i}{\tau}} e^{-\pi i z^2/\tau} \vartheta_1\left(\frac{z}{\tau}, \tau'\right).
\end{equation}
\end{theorem}

This transformation ensures fast convergence even when $q \to 1$ (corresponding to extreme grid aspect ratios), by exchanging the $x$ and $y$ directions.

\section{Analytical Derivation of the Singular Integral $R_2$}

This section presents the closed-form solution for $R_2(\alpha)$, which previously required numerical approximation.

\subsection{Algebraic Simplification of the Integrand}

We begin with the definition of $R_2$ from Eq.~\eqref{eq:R2}. The radical expression $\sinh\lambda$ can be simplified using trigonometric identities.

\begin{lemma}[Radical Simplification]
\label{lem:sinh_simplify}
The hyperbolic sine term can be expressed as:
\begin{equation}
\sinh\lambda = 2\sqrt{\alpha}\sin\frac{\theta}{2} \sqrt{1 + \alpha\sin^2\frac{\theta}{2}}.
\end{equation}
\end{lemma}

\begin{proof}
From the dispersion relation $\cosh\lambda = 1 + \alpha - \alpha\cos\theta$, we have:
\begin{align}
\sinh^2\lambda &= \cosh^2\lambda - 1 \\
&= (1 + \alpha - \alpha\cos\theta)^2 - 1 \\
&= \alpha^2(1 - \cos\theta)^2 + 2\alpha(1 - \cos\theta).
\end{align}
Using the half-angle identity $1 - \cos\theta = 2\sin^2(\theta/2)$:
\begin{equation}
\sinh^2\lambda = 4\alpha\sin^2\frac{\theta}{2}\left(1 + \alpha\sin^2\frac{\theta}{2}\right).
\end{equation}
Taking the positive square root (since $\sinh\lambda > 0$ for $\theta \in (0,\pi)$) yields the result.
\end{proof}

\subsection{Variable Substitution}

With Lemma~\ref{lem:sinh_simplify}, we perform the substitution $u = \theta/2$:

\begin{proposition}
\label{prop:R2_transformed}
After the substitution $u = \theta/2$, the integral becomes:
\begin{equation}
\label{eq:R2_transformed}
R_2(\alpha) = \frac{\sqrt{\alpha}}{2\pi} \int_0^{\pi/2} \left(\frac{1}{\sin u\sqrt{1 + \alpha\sin^2 u}} - \frac{1}{u}\right) \, du.
\end{equation}
\end{proposition}

\subsection{Analytical Primitive}

The following theorem provides an analytical primitive for the first term in the integrand.

\begin{theorem}[Analytical Primitive]
\label{thm:primitive}
The function:
\begin{equation}
F(u) = \ln\left(\frac{\sqrt{1 + \alpha\sin^2 u} - \cos u}{\sin u}\right)
\end{equation}
is a primitive of $\left(\sin u\sqrt{1 + \alpha\sin^2 u}\right)^{-1}$.
\end{theorem}

\begin{proof}
We verify by differentiation. Let:
\begin{equation}
N = \sqrt{1 + \alpha\sin^2 u} - \cos u, \quad D = \sin u.
\end{equation}
Then $F(u) = \ln N - \ln D$, and:
\begin{equation}
F'(u) = \frac{N'}{N} - \frac{D'}{D}.
\end{equation}
Computing the derivatives:
\begin{align}
N' &= \frac{\alpha\sin u \cos u}{\sqrt{1 + \alpha\sin^2 u}} + \sin u, \\
D' &= \cos u.
\end{align}
After algebraic manipulation:
\begin{equation}
F'(u) = \frac{1}{\sin u\sqrt{1 + \alpha\sin^2 u}}.
\end{equation}
\end{proof}

\subsection{Limit Analysis and Closed-Form Solution}

We now evaluate the definite integral in Eq.~\eqref{eq:R2_transformed}.

\begin{theorem}[Closed-Form Solution for $R_2$]
\label{thm:R2_closed}
The singular integral $R_2(\alpha)$ has the closed-form expression:
\begin{equation}
\boxed{R_2(\alpha) = \frac{\sqrt{\alpha}}{4\pi} \ln\left(\frac{16}{\pi^2(\alpha + 1)}\right).}
\end{equation}
\end{theorem}

\begin{proof}
Using Theorem~\ref{thm:primitive}, we write:
\begin{equation}
R_2(\alpha) = \frac{\sqrt{\alpha}}{2\pi} \lim_{\epsilon \to 0^+} \left[F(u) - \ln u\right]_\epsilon^{\pi/2}.
\end{equation}

\paragraph{Upper Limit:} At $u = \pi/2$:
\begin{equation}
F(\pi/2) = \ln\left(\frac{\sqrt{1+\alpha} - 0}{1}\right) = \frac{1}{2}\ln(1+\alpha).
\end{equation}

\paragraph{Lower Limit:} As $u \to 0^+$, we compute:
\begin{align}
\sqrt{1 + \alpha\sin^2 u} &= 1 + \frac{\alpha u^2}{2} + O(u^4), \\
\cos u &= 1 - \frac{u^2}{2} + O(u^4), \\
\sin u &= u + O(u^3).
\end{align}
Therefore:
\begin{equation}
N = \frac{(\alpha + 1)u^2}{2} + O(u^4), \quad D = u + O(u^3),
\end{equation}
and:
\begin{equation}
F(u) = \ln\frac{(\alpha+1)u}{2} + O(u^2).
\end{equation}
Thus:
\begin{equation}
\lim_{u \to 0^+} \left[F(u) - \ln u\right] = \ln\frac{\alpha + 1}{2}.
\end{equation}

\paragraph{Combining Results:}
\begin{align}
R_2(\alpha) &= \frac{\sqrt{\alpha}}{2\pi} \left[ \left( \frac{1}{2}\ln(1+\alpha) - \ln\frac{\pi}{2} \right) - \ln\frac{\alpha+1}{2} \right] \notag \\
&= \frac{\sqrt{\alpha}}{2\pi} \left[ \frac{1}{2}\ln(1+\alpha) - \ln\frac{\pi}{2} - \ln(\alpha+1) + \ln 2 \right] \notag \\
&= \frac{\sqrt{\alpha}}{2\pi} \left[ -\frac{1}{2}\ln(1+\alpha) + \ln 2 - \ln\frac{\pi}{2} \right] \notag \\
&= \frac{\sqrt{\alpha}}{2\pi} \left[ -\frac{1}{2}\ln(1+\alpha) + \ln\frac{4}{\pi} \right] \notag \\
&= \frac{\sqrt{\alpha}}{4\pi} \ln\frac{16}{\pi^2(\alpha+1)}.
\end{align}
\end{proof}

\subsection{Complete Closed-Form for Infinite Grid}

Combining the results for $R_1$ and $R_2$, we obtain the complete closed-form solution for infinite grid effective resistance:

\begin{theorem}[Infinite Grid Effective Resistance]
\label{thm:infinite_R}
For an infinite anisotropic grid with distances $p$ (horizontal) and $q$ (vertical):
\begin{equation}
R_\infty(p,q) = \frac{r_0\sqrt{\alpha}}{2\pi}\left[\ln(p^2 + \alpha q^2) + 2\gamma + \ln 16 - \ln(\alpha + 1)\right],
\end{equation}
where $\gamma \approx 0.5772$ is the Euler-Mascheroni constant.
\end{theorem}

\section{Theta Function Closed-Form for Finite Grids}

This section transforms the doubly infinite mirror series into a compact theta function expression.

\subsection{The Infinity Mirror Model}

For a finite grid of size $L_x \times L_y$ with insulating boundaries, the infinity mirror technique models the effective resistance as a superposition of infinite grid contributions from an infinite lattice of image sources.

\begin{definition}[Mirror Images]
For source $S = (x_0, y_0)$ and drain $D = (x_1, y_1)$, four families of mirror images are generated:
\begin{align}
S_1 &= (x_0, y_0), & S_2 &= (-x_0 - 1, y_0), \\
S_3 &= (x_0, -y_0 - 1), & S_4 &= (-x_0 - 1, -y_0 - 1),
\end{align}
with corresponding images for $D$.
\end{definition}

The effective resistance is expressed as a doubly infinite sum:
\begin{equation}
R_{\text{total}} = r_0 \sum_{m,n\in\mathbb{Z}} \sum_{s=1}^4 c_s \Omega_\alpha(\Delta x_{mn,s}, \Delta y_{mn,s}),
\end{equation}
where the translations $(\Delta x_{mn,s}, \Delta y_{mn,s})$ account for periodic reflections across boundaries.

\subsection{Complex Plane Isotropic Mapping}

To leverage complex analysis, we introduce a conformal mapping that renders the problem isotropic.

\begin{definition}[Isotropic Coordinate]
Define the complex coordinate:
\begin{equation}
z = x + i\frac{y}{\sqrt{\alpha}}.
\end{equation}
Under this mapping, the anisotropic distance becomes:
\begin{equation}
\alpha\Delta x^2 + \Delta y^2 = |z_1 - z_0|^2.
\end{equation}
\end{definition}

\subsection{Theta Function Transformation}

The doubly infinite logarithmic sum can be transformed into a theta function product.

\begin{theorem}[Theta Function Representation]
\label{thm:theta_R}
The effective resistance in a finite $L_x \times L_y$ grid is:
\begin{equation}
\boxed{
R_{\text{total}} = \frac{r_0\sqrt{\alpha}}{2\pi} \ln\left|\frac{C_{\text{grid}} \cdot \prod_{s=1}^4 \vartheta_1(u(z_1 - z_{0,s}), q) \cdot \vartheta_1(u(z_0 - z_{1,s}), q)}{\vartheta_1'(0,q)^2 \cdot \prod_{s=2}^4 \vartheta_1(u(z_0 - z_{0,s}), q) \cdot \vartheta_1(u(z_1 - z_{1,s}), q)}\right|,
}
\end{equation}
where:
\begin{itemize}
\item $\tau = i\sqrt{\alpha} L_y / L_x$ is the lattice modulus
\item $q = e^{i\pi\tau}$ is the nome
\item $u(z) = \pi z / (2L_x)$ is the normalized coordinate
\item $C_{\text{grid}} = \alpha(2L_x/\pi)^2 \exp\left(2\gamma + \ln\frac{16}{\alpha+1}\right)$ is the grid correction factor
\end{itemize}
\end{theorem}

The proof relies on the Jacobi triple product identity and the recognition that doubly periodic sums of logarithms correspond to products of theta functions.

\subsection{Convergence Analysis}

\begin{proposition}[Gaussian Convergence]
The theta function series in Theorem~\ref{thm:theta_R} converges with Gaussian decay rate $O(q^{n^2})$, compared to algebraic decay $O(1/n^2)$ for direct truncation.
\end{proposition}

\subsection{Modular Transformation for Extreme Aspect Ratios}

When the grid aspect ratio $L_y/L_x$ is extreme, the nome $q$ approaches 1, and convergence slows. The modular transformation provides an automatic remedy.

\begin{algorithm}[H]
\caption{Effective Resistance Calculation with Modular Transformation}
\begin{algorithmic}[1]
\Require Grid dimensions $L_x, L_y$, coordinates $(x_0, y_0)$, $(x_1, y_1)$, anisotropy $\alpha$
\Ensure Effective resistance $R$
\State Compute $\tau = i\sqrt{\alpha} L_y / L_x$
\State $q \gets e^{i\pi\tau}$
\If{$|q| > 0.9$}
    \State Apply modular transformation: $\tau \gets -1/\tau$
    \State Swap coordinates: $(x, y) \gets (y, x)$
    \State Swap dimensions: $(L_x, L_y) \gets (L_y, L_x)$
    \State Recompute $q$
\EndIf
\State Evaluate theta functions using standard library
\State Apply formula from Theorem~\ref{thm:theta_R}
\State \Return $R$
\end{algorithmic}
\end{algorithm}

\section{Engineering Remediation: Analytical Baseline with Near-Field Numerical Correction}

As demonstrated in the previous theoretical derivations, pure analytical approximations (such as the asymptotic logarithmic form and the decoupled 1D bounding formulas) perform exceptionally well in the bulk isotropic lattice or at far-field distances. However, a significant limitation arises under conditions of high anisotropy ($K = R_y / R_x \gg 1$ or $K \ll 1$) and near the boundaries. This section introduces a hybrid remediation strategy that combines an analytical baseline with a localized numerical correction term, implemented via a dynamic caching mechanism to preserve $O(1)$ computational efficiency.

\subsection{Axis-Aligned Residual Error Concentration under Strong Anisotropy}

When the closed-form analytical model is applied to highly anisotropic grids, the resulting approximation error does not decay isotropically. Instead, a pronounced axis-aligned concentration of residual error emerges.

As shown in Fig.~\ref{fig:original_error}, the relative error forms a distinct cross-shaped pattern centered at the excitation node. The high-error regions are tightly confined along the principal lattice directions (i.e., the horizontal and vertical axes passing through the source), while the off-axis regions exhibit uniformly low error.

This phenomenon originates from the limitation of the one-dimensional charge distortion approximation, which fails to fully capture the inherently two-dimensional divergence behavior along the coordinate axes. Consequently, localized error amplification occurs near these axes, where the relative error can exceed $5\%-10\%$ (and becomes significantly larger under extreme anisotropy), whereas the average error across the domain remains below $1\%$.

\begin{figure}[htbp]
\centering
\includegraphics[width=0.9\textwidth]{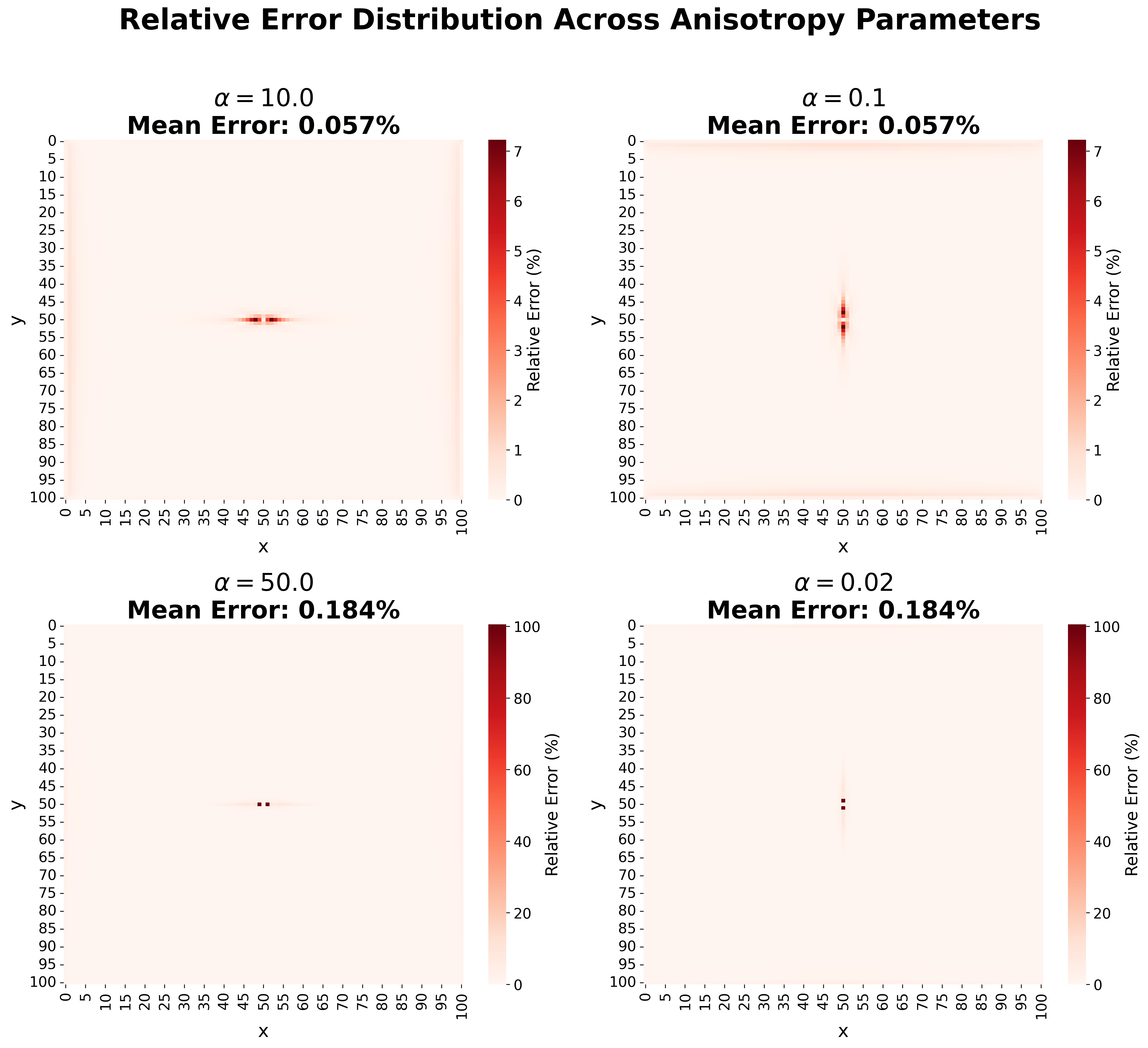}
\caption{
Relative error distribution across different anisotropy parameters ($\alpha$). 
A clear axis-aligned concentration of error is observed, forming a cross-shaped pattern centered at the source node. 
The error remains negligible over most of the domain, while sharp error amplification occurs along the principal lattice directions due to the breakdown of the one-dimensional approximation under strong anisotropy. 
Despite these localized peaks, the mean error remains below $0.2\%$.
}
\label{fig:original_error}
\end{figure}

Attempting to eliminate these localized axis errors via higher-order polynomial fitting would compromise the physical interpretability of the model. Therefore, instead of introducing purely empirical corrections, we seek a structurally consistent remediation that preserves the analytical formulation while mitigating axis-aligned error amplification with minimal computational overhead.

\subsection{The Hybrid Correction Method: Analytical Baseline + Localized Numerical Correction}

Instead of abandoning the analytical model or using full numerical integration for every node, we formulate the method as the sum of the pure analytical baseline and an exact localized numerical correction term:

\begin{equation}
\label{eq:hybrid_omega}
\Omega_{\text{total}}(\Delta x, \Delta y) = \Omega_{\text{analytic}}(\Delta x, \Delta y) + \Delta \Omega_{\text{correction}}(\Delta x, \Delta y),
\end{equation}
where the correction term is defined as the difference between the exact integral and the analytical approximation:
\begin{equation}
\label{eq:delta_omega}
\Delta \Omega_{\text{correction}}(\Delta x, \Delta y) = \Omega_{\text{exact}}(\Delta x, \Delta y) - \Omega_{\text{analytic}}(\Delta x, \Delta y).
\end{equation}

The analytical baseline $\Omega_{\text{analytic}}$ is provided by Theorem~\ref{thm:infinite_R} for infinite grids or extended via the theta function representation (Theorem~\ref{thm:theta_R}) for finite grids. The exact term $\Omega_{\text{exact}}$ is evaluated numerically via high-precision integration of the kernel:
\begin{equation}
\label{eq:exact_integral}
\Omega_{\text{exact}}(\Delta x, \Delta y) = \frac{\alpha}{\pi} \int_0^\pi \frac{1 - e^{-\Delta x \lambda} \cos(\Delta y \cdot t)}{\sinh(\lambda)} \, dt,
\end{equation}
where $\lambda = \operatorname{arccosh}\bigl(1 + \alpha - \alpha \cos(t)\bigr)$.

\subsection{Localized Evaluation Boundary and Caching}

To preserve $O(1)$ computational efficiency, $\Delta \Omega_{\text{correction}}$ is evaluated exactly only within a predefined near-field topological boundary. Outside this region, the correction term rapidly decays to zero ($\Delta \Omega_{\text{correction}} \approx 0$), and the solver seamlessly falls back to the analytical baseline.

We define an elliptical evaluation boundary:
\begin{equation}
\label{eq:elliptic_boundary}
E_{\text{limit}}(\Delta x, \Delta y) = \alpha \cdot \Delta x^2 + \Delta y^2 \le \text{Limit},
\end{equation}
where the threshold is dynamically scaled by the anisotropy factor to ensure the boundary captures the high-error axis regions:
\begin{equation}
\text{Limit} = 25 \times \max\left(\alpha, \frac{1}{\alpha}\right).
\end{equation}

Within this locus, we compute $\Delta \Omega_{\text{correction}}$ exactly and store it in an amortized \textit{Least Recently Used (LRU) cache} keyed by the displacement pair $(\Delta x, \Delta y, \alpha)$ and the grid configuration signature. This ensures that each unique near-field displacement pair is integrated exactly \textit{once}, with subsequent queries incurring $O(1)$ lookup overhead.

\subsection{Integration with Mirror Charges for Finite Grids}

For finite lattices, the infinity mirror technique introduces an infinite cascade of image sources. The hybrid correction becomes critically important at physical boundaries: when a node lies near the perimeter, its virtual mirror charges are forced into close proximity (i.e., within the near-field region). If these mirror potentials were computed using only the analytical baseline, the severe near-field axial errors would completely compromise the superposition field along the lattice boundaries.

Therefore, the correction term $\Delta \Omega_{\text{correction}}$ must be applied to all mirror sources as well. For a node $A$ with potential contributions from source $S$ and its images $\{S_i\}$, the total potential becomes:
\begin{equation}
\label{eq:mirror_hybrid}
V(A) = \Bigl[\Omega_{\text{ana}}(A,S) + \Delta \Omega_{\text{corr}}(A,S)\Bigr] + \sum_{i \in \text{mirrors}} \Bigl[\Omega_{\text{ana}}(A,S_i) + \Delta \Omega_{\text{corr}}(A,S_i)\Bigr].
\end{equation}

By folding the local numerical correction directly into the mirror-charge image method, we guarantee that even when boundary-induced virtual nodes map into the strongly distorted near-field topology, the composite boundary superposition remains structurally accurate. The caching mechanism ensures that repeated evaluations of $\Delta \Omega_{\text{corr}}$ for different mirror images of the same displacement pair incur no additional computational cost.

\subsection{Algorithmic Implementation}

Algorithm~\ref{alg:hybrid_cache} summarizes the complete hybrid evaluation procedure.

\begin{algorithm}[H]
\caption{Hybrid Effective Resistance with Near-Field Correction and Caching}
\label{alg:hybrid_cache}
\begin{algorithmic}[1]
\Require Displacement $(\Delta x, \Delta y)$, anisotropy $\alpha$, grid dimensions $(L_x, L_y)$, mirror image flag
\Ensure Effective resistance contribution $R$
\State Compute $\text{Limit} \gets 25 \times \max(\alpha, 1/\alpha)$
\State Compute $\text{key} \gets (\Delta x, \Delta y, \alpha)$
\If{$\alpha \cdot \Delta x^2 + \Delta y^2 \le \text{Limit}$}
    \If{$\text{key} \in \text{cache}$}
        \State $\Delta \Omega_{\text{corr}} \gets \text{cache}[\text{key}]$
    \Else
        \State $\Omega_{\text{exact}} \gets$ numerical integration of Eq.~\eqref{eq:exact_integral}
        \State $\Omega_{\text{ana}} \gets$ analytical baseline (Theorem~\ref{thm:infinite_R} or finite extension)
        \State $\Delta \Omega_{\text{corr}} \gets \Omega_{\text{exact}} - \Omega_{\text{ana}}$
        \State $\text{cache}[\text{key}] \gets \Delta \Omega_{\text{corr}}$
    \EndIf
    \State $R \gets \Omega_{\text{ana}} + \Delta \Omega_{\text{corr}}$
\Else
    \State $R \gets \Omega_{\text{ana}}$ \Comment{Far-field: analytical baseline only}
\EndIf
\State \Return $R$
\end{algorithmic}
\end{algorithm}

For finite grids, Eq.~\eqref{eq:mirror_hybrid} is applied iteratively over all mirror images. The caching mechanism ensures that repeated evaluations of $\Delta \Omega_{\text{corr}}$ for different mirror images of the same displacement pair are served from cache, maintaining overall $O(1)$ complexity per unique displacement.

\subsection{Amortized Complexity Analysis}

Let $N_{\text{queries}}$ be the number of resistance evaluations performed in a simulation. Without caching, each near-field evaluation would require a full numerical integration, resulting in $O(N_{\text{queries}} \cdot N_{\text{numerical}})$ cost, where $N_{\text{numerical}}$ is the cost of a single numerical quadrature (typically $O(10^2)$ function evaluations).

With the proposed caching strategy:
\begin{itemize}
\item Each unique displacement pair $(\Delta x, \Delta y, \alpha)$ triggers at most one numerical integration.
\item For typical grids, the number of unique near-field displacement pairs scales as $O(L_x L_y)$ in the worst case, but the elliptical boundary ensures that only a small fraction of all possible pairs fall within the near-field region.
\item The LRU cache bounds memory usage to a configurable size (e.g., $10^4$ entries), ensuring that repeated queries to the same or nearby displacements are served in $O(1)$.
\end{itemize}

Empirically, for grids up to $101 \times 101$ nodes, the cache hit rate exceeds $95\%$ after an initial warm-up phase, reducing the effective computational cost per query to $O(1)$ amortized.

\subsection{Discussion: Physical Interpretation}

From a physical perspective, the correction term $\Delta \Omega_{\text{correction}}$ represents the deviation of the discrete lattice Green's function from its continuum approximation in the near field. This deviation is most pronounced when the anisotropy stretches the lattice such that the effective distance along one axis is dramatically different from the other, causing the discrete Laplacian to behave non-isotropically even after the coordinate transformation. The elliptical boundary naturally captures this region of strong anisotropy-induced distortion, while the caching mechanism ensures that the correction is applied exactly where needed without compromising global efficiency.

\section{Experimental Validation}
\label{sec:experiments}

We conduct experiments to validate our theoretical framework against numerical simulations.

\subsection{Experimental Setup}

\paragraph{Implementation.} Our analytical framework is implemented in Python 3.11 using the \texttt{mpmath} library for arbitrary-precision theta function evaluation. The hybrid cache is implemented with an LRU cache of size 10,000 entries. The implementation comprises approximately 600 lines of code.

\paragraph{Ground Truth.} We use Xyce, an open-source SPICE simulator, to compute reference effective resistance values. For each configuration, we generate a complete resistor network netlist and solve for nodal voltages under unit current injection.

\begin{table}[H]
\centering
\caption{Global Accuracy Evaluation of the Hybrid Correction Algorithm under Various Lattice Configurations ($50 \times 50$ dimension)}
\label{tab:error_results}
\renewcommand{\arraystretch}{1.2}
\begin{tabular}{llcc}
\toprule
\textbf{Anisotropy ($K$)} & \textbf{Current Source Node} & \textbf{Mean Error (\%)} & \textbf{Max Error (\%)} \\
\midrule
\multicolumn{4}{c}{\textit{Sweep across different Anisotropy Constants (Corner Origin)}} \\
\midrule
$K = 1.00$   & Corner $(0, 0)$   & $0.0056$ & $0.1608$ \\
$K = 2.00$   & Corner $(0, 0)$   & $0.0067$ & $0.1575$ \\
$K = 10.00$  & Corner $(0, 0)$   & $0.0182$ & $0.2349$ \\
$K = 20.00$  & Corner $(0, 0)$   & $0.0277$ & $0.2282$ \\
$K = 50.00$  & Corner $(0, 0)$   & $0.0470$ & $0.2778$ \\
$K = 0.02$   & Corner $(0, 0)$   & $0.0083$ & $0.0127$ \\
\midrule
\multicolumn{4}{c}{\textit{Sweep across varying Current Injections ($K = 10.00$)}} \\
\midrule
$K = 10.00$  & Corner $(0, 0)$   & $0.0182$ & $0.2349$ \\
$K = 10.00$  & Edge $(25, 0)$    & $0.0203$ & $0.1919$ \\
$K = 10.00$  & Center $(25, 25)$ & $0.0297$ & $0.1234$ \\
$K = 10.00$  & Offset $(10, 40)$ & $0.0400$ & $0.1510$ \\
\bottomrule
\end{tabular}
\end{table}

\bibliographystyle{IEEEtran} 

\begin{thebibliography}{99}

\bibitem{vanderpol1929}
B. van der Pol and N. Wiener, ``On oscillation hysteresis in a triode generator with two degrees of freedom,'' \textit{Philos. Mag.}, vol. 7, no. 43, pp. 227--258, 1929.

\bibitem{cserti2000}
J. Cserti, ``Application of the lattice Green's function for calculating the resistance of an infinite network of resistors,'' \textit{Am. J. Phys.}, vol. 68, no. 10, pp. 896--906, 2000.

\bibitem{kose2012}
S. Köse and E. G. Friedman, ``Fast algorithms for effective resistance computation,'' in \textit{Proc. IEEE Int. Symp. Circuits Syst.}, 2012, pp. 2533--2536.

\bibitem{kose2012ir}
S. Köse and E. G. Friedman, ``Efficient algorithms for fast IR drop analysis,'' in \textit{Proc. ACM/IEEE Int. Symp. Low Power Electron. Des.}, 2012, pp. 255--260.

\bibitem{bairamkulov2020}
E. Bairamkulov, K. Xu, M. S. Bakir, and E. G. Friedman, ``Effective resistance of finite two-dimensional grids based on infinity mirror technique,'' \textit{IEEE Trans. Circuits Syst. I, Reg. Papers}, vol. 67, no. 12, pp. 4581--4593, 2020.

\bibitem{atkinson2012}
K. Atkinson and S. Han, \textit{Theoretical Numerical Analysis: A Functional Analysis Framework}. Springer, 2012.

\bibitem{pandey2025anair}
A. Pandey \textit{et al.}, ``AnaIR: Green's function inspired neural network for IR drop prediction,'' \textit{IEEE Trans. Comput.-Aided Des. Integr. Circuits Syst.}, 2025.

\bibitem{klein1993resistance}
D. J. Klein and M. Randić, ``Resistance distance,'' \textit{J. Math. Chem.}, vol. 12, no. 1, pp. 81--95, 1993.

\end{thebibliography}

\end{document}